\documentclass[onecolumn]{IEEEtran}
\IEEEoverridecommandlockouts

\usepackage{amsfonts}
\usepackage{amsmath}
\usepackage{amssymb}
\usepackage{algorithm,algorithmic}
\usepackage{array}
\usepackage{textcomp}
\usepackage{stfloats}
\usepackage{xspace}
\usepackage{caption}
\usepackage{multirow}
\usepackage{tikz}
\usepackage{pifont}
\usepackage{inconsolata}
\usepackage[export]{adjustbox}
\usepackage{flushend}
\usepackage{lipsum}
\usepackage{dsfont}
\usepackage{scalefnt}
\usepackage{fontawesome5}

\usepackage[utf8]{inputenc} 
\usepackage[T1]{fontenc}
\usepackage[scaled=.85]{beramono}

\usepackage{setspace}
\usepackage{mathtools}
\usepackage{amsthm}
\usepackage{url}
\usepackage{verbatim}
\usepackage{graphicx}
\usepackage{euscript}
\usepackage{cite}
\usepackage{soul}
\usepackage{mathrsfs,xspace,multicol}
\usepackage[usenames,dvipsnames,table]{xcolor}
\usepackage{tabularx}
\usepackage{booktabs}
\usepackage[normalem]{ulem}
\usepackage[inline, shortlabels]{enumitem}
\setlist{itemjoin ={,\enspace},itemjoin* = {, and\enspace}}

\usepackage[font=footnotesize]{caption}
\usepackage{hyperref}
\hypersetup{
	colorlinks=true,
	linkcolor=black,
	filecolor=black,      
	urlcolor=black,
	citecolor =black,
}

\newtheorem{theorem}{Theorem}

\def\khz								  {\texttt{kHz}\xspace}

\def\g                                      {g}

\def\l										{\left(}
\def\r								       {\right)}

\def\Z									  {\mathbb Z}
\def\R									  {\mathbb R}

\def\DE					                {\stackrel{\rm{def}}{=}}
\def\eg					                {\emph{e.g.\xspace}}
\def\ie					         	      {\emph{i.e.\xspace}}

\def\usf					   		   {USF\xspace}
\def\g               				 	 {g}
\def\HDR							{HDR\xspace}

\def\Dis					    {\mathscr{Q}}

\newcommand{\norm}[1]           {\left\lVert {#1} \right\rVert }
\newcommand{\fig}[1]			    {Fig.~\ref{#1}}

\newcommand{\id}[1]					{\mathbb{I}_{#1} }
\newcommand{\sqb}[1]		    {\left[ #1 \right]}

\newcommand\rob[1]				  {\l #1 \r}

\newcommand\tabref[1]		    {Table \ref{#1}}

\newcommand\thmref[1]		  {Theorem \ref{#1}}

\newcommand{\bpara}[1]		 {\smallskip \noindent {\textbf{#1}}}

\newcommand{\Lp}[1]            {{\mathbf{L}}_{{#1}}}
\newcommand{\lp}[1]             {{\ell}_{#1}}

\newcommand{\abs}[1]			{\left| #1\right|}

\newcommand\subfig[3]	    {Fig.~\ref{#1} (${\mathsf{#2}}_{#3}$)}
\newcommand{\fe}[1]				{\left[\kern-0.30em\left[#1\right]\kern-0.30em\right]}
\newcommand{\ft}[1]				{\left[\kern-0.15em\left[#1\right]\kern-0.15em\right]}
\newcommand{\flr}[1]			{\left\lfloor #1 \right\rfloor}

\newcommand{\MOL}[2]	   {\mathscr{M}_{#1} ({#2} )}

\renewcommand\bar\underline

\newcommand{\QO}[2]		  {\Dis_{#1} \rob{#2} }

\usepackage[most,breakable]{tcolorbox}

\tcbset{
	left = 0pt,
	top = 0pt,
	bottom = 0pt,
	breakable,
	before skip=0.2cm,
	after ={\parindent0em},
	colback=yellow!10!white, colframe=red!50!black, 
	highlight math style= {enhanced, 
		colframe=red,colback=red!10!white,
		boxsep=0pt,
}}

\newtcbox{\abox}[1][brown]{on line,
	arc=0pt,
	colback=#1!10!white,
	colframe=#1!50!black,
	arc=0pt,
	outer arc=0pt,
	top=1pt,
	bottom=0.5pt,
	left=0mm,
	right=0mm,
	leftrule=0pt,
	rightrule=0pt,
	toprule=0.3mm,
	bottomrule=0.3mm,
	boxsep=0.1mm
}

\makeatletter
\def\moverlay{\mathpalette\mov@rlay}
\def\mov@rlay#1#2{\leavevmode\vtop{
		\baselineskip\z@skip \lineskiplimit-\maxdimen
		\ialign{\hfil$\m@th#1##$\hfil\cr#2\crcr}}}
\newcommand{\charfusion}[3][\mathord]{
	#1{\ifx#1\mathop\vphantom{#2}\fi
		\mathpalette\mov@rlay{#2\cr#3}
	}
	\ifx#1\mathop\expandafter\displaylimits\fi}
\makeatother
\makeatletter
\newcommand*{\rom}[1]{\expandafter\@slowromancap\romannumeral #1@}
\makeatother

\def\qn							    {q}

\def\Td							   {T_d}
\def\lam   						  {\lambda}			
\def\glcm					    {\lambda_{\mathsf{lcm}}}
\def\tauL						 {\tau_{n} }

\def\Lc							  {{c}}
\def\B							   {\mathbf{B}}
\def\mgn					  {\mathbf{\g}_{n}}
\def\tmgn					 {\mathbf{\widetilde{\g}}_{n}}
\def\myn					  {\mathbf{y}_{n}}

\def\mgamman	     {\boldsymbol{\gamma}_{n}}
\def\tmgamman	    {\boldsymbol{\widetilde{\gamma}}_{n}}

\def\mL							{\boldsymbol{\Lambda}}
\def\mK							{\boldsymbol{\kappa}}
\def\hn							 {h_{n}}
\def\Gg							{\mathscr{G}_{\g,n}}

\def\Tc							{T_c^{n}}
\def\Tn							{T_n}
\def\Tll					     {T_l}
\def\tc							 {t_c^{n}}
\def\mG						 {\boldsymbol{\Gamma}}

\newcommand{\BL}[2]							 {#1 \in \mathcal{B}_{#2}}
\newcommand{\yl}[1]							   {y_{#1}}

\newcommand{\yln}[1]					     {y_{#1}^{n}}
\newcommand{\gln}[1]					     {\g_{#1}^{n}}

\renewcommand{\gcd}[1]					 {\mathsf{gcd}\rob{#1}}
\newcommand{\lcm}[1]					   {\mathsf{lcm}\rob{#1}}
\newcommand{\qnl}[1]					     {\qn_{#1}}
\newcommand{\gl}[1]					          {\g_{#1}}

\newcommand{\gammal}[1]			     {\gamma_{#1}}

\newcommand{\gammaln}[1]		    {\gamma_{#1}^{n}}

\newcommand{\kappal}[1]			        {\kappa_{#1}}
\newcommand{\Gammal}[1]			     {\Gamma_{#1}}
\newcommand{\Dg}[1]						    {\eta_{#1}}
\newcommand{\wl}[1]							 {\eta_{#1}}

\newcommand{\Tl}[1]							 {T_{#1}^{n}}

\definecolor{rowgray}{gray}{0.94}
\definecolor{LightGrayCell}{gray}{0.90} % if not already defined

\begin{document}

\title{Asynchronous Multi-Channel USF: \\
Modified CRT for Modulo Unfolding
\\
\thanks{This work is supported by the European Research Council's Starting Grant for ``CoSI-Fold'' (101166158). Further details on {Unlimited Sensing} and materials on \textit{reproducible research} are available via  \href{https://bit.ly/USF-Link}{\texttt{https://bit.ly/USF-Link}}.}
}

\author{ 
Ruiming Guo 
and
Ayush Bhandari \\ 
Dept. of Electrical and Electronic Engg., 
{Imperial College London}, SW7 2AZ, UK. \\ 
\texttt{ruiming.guo@imperial.ac.uk} $\bullet$
\texttt{a.bhandari@imperial.ac.uk}
}

\maketitle

\vspace{1mm}

{

\centering
\sf
\color{blue} To appear in the proceedings of 2026 European Signal Processing Conference (EUSIPCO).

}

\vspace{1cm}

\begin{abstract}
The Unlimited Sampling Framework (USF) overcomes the traditional trade‑off between dynamic range and digital resolution, achieving performance unattainable with standard ADCs. Its multi‑channel extension (MC‑USF) enables reconstruction from multiple folded measurements at critical sampling rates. Existing MC‑USF methods typically rely on Chinese Remainder Theorem (CRT)-based unfolding, which requires strict channel‑level sampling synchronization and is therefore vulnerable to timing mismatch, jitter, and drift. This paper introduces an asynchronous MC‑USF architecture that eliminates the need for synchronization. By viewing spatial–temporal signal lifting as inducing smoothness over a graph of sensing channels, we develop a reconstruction strategy robust to temporal misalignment. Numerical experiments validate the approach, demonstrating accurate recovery and enabling more practical multi‑channel USF implementations.

\end{abstract}


\tableofcontents

\newpage

\spacing{1.5}

\section{Introduction}

Unlimited Sensing Framework (\usf) \cite{Bhandari:2017:C,Bhandari:2020:Ja,Bhandari:2021:J} is a novel acquisition paradigm that rethinks conventional sampling pipeline by advocating controlled non-linearities at the front end.
USF intentionally operates signal folding in the analog domain and then reconstructs the signal digitally by exploiting structure in the underlying waveform. 
This approach enables sensing systems that simultaneously achieve high-dynamic-range and high-digital-resolution, beyond the fundamental limits imposed by classical linear analog-to-digital (ADC) architectures.
Shifting from single-channel acquisition \cite{Bhandari:2017:C,Bhandari:2020:Ja,Bhandari:2021:J,Guo:2023:C}, a natural extension of USF is multi-channel (MC) sensing \cite{Guo:2023:Ca,Guo:2024:J}, where the same underlying signal is observed through multiple, shifted views. 
Specifically, consider a bandlimited signal $\g\rob{t}$. We observe samples of the form $\gln{l} = \g\rob{ \Tn + \Tll }$, where $\{\Tn, \Tll \}$ represent temporal sampling instants and time-varying channel-dependent offsets. 
Typical examples include: $\Tn = n T, n\in\Z$ and $\Tll = l S$ may correspond to spatial separation in an array \cite{FernandezMenduina:2021:J}, time delays ($S = \Td$) in a multi-coset architecture \cite{Guo:2025:C,Guo:2025:Ca} (see \subfig{fig:demo}{b}{}), or other forms of structured diversity.     

The key insight is that the $L$-channels do not observe independent signals; rather, they capture different slices of the same evolution of $\g\rob{t}$. 
This space-time redundancy introduces strong structural constraints that can be employed algorithmically, enabling recovery at sampling rates below Nyquist or with relaxed hardware requirements. 
Classic examples include sub-Nyquist recovery of multiband \cite{Guo:2025:C,Guo:2025:Ca} or sum-of-sinusoids signals \cite{Guo:2024:J}, as well as recovery methods based on the Chinese Remainder Theorem (CRT) \cite{Li:2009:J,Wang:2010:J,Gong:2021:J,Xiao:2024:J,Guo:2025:J,Yan:2026:J}.   

CRT-based MC-\usf has received particular attention in recent years. 
For a fixed time index, collecting folded samples across multiple channels reduces to solving a system of congruence equations.
This provides two important benefits: reduced sampling rates and relaxed quantization requirements.
Together, these advantages translate into improved data efficiency and lower-power hardware implementations. 

\begin{figure}[!t]
\centering	\includegraphics[width=0.6\linewidth]{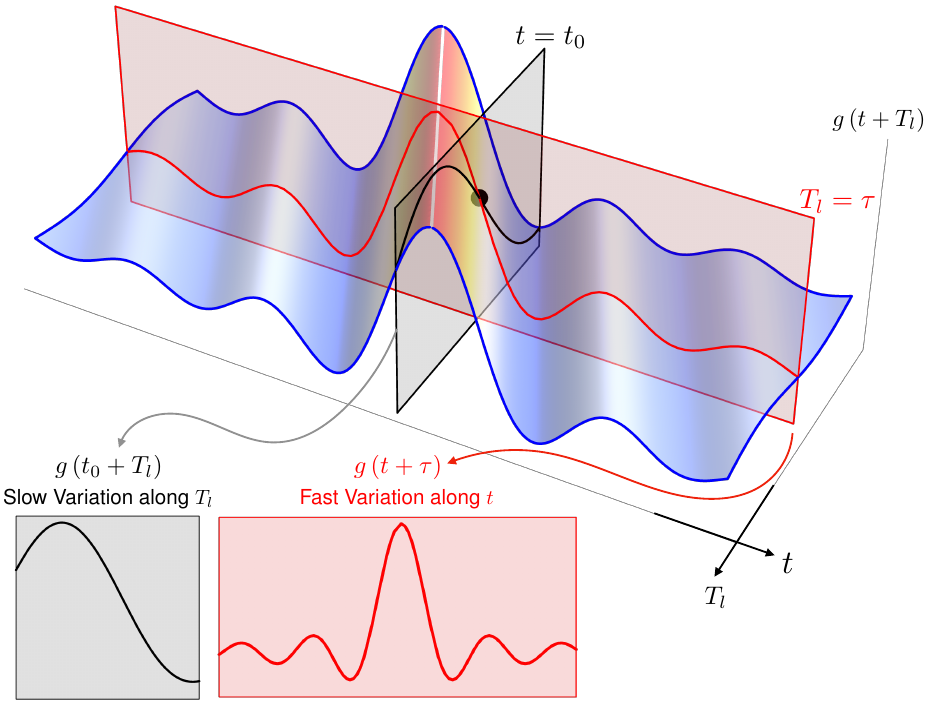}
\caption{
Lifting $g(t) = \operatorname{sinc}(t)$ to $\g\rob{t+\Tll}$ and the corresponding 2D geometry. Depending on the ratio between $t$ and $\Tll$, sampling along $\Tll$ captures a slower variation compared to temporal evolution along $t$.} 
\label{fig:Asyn}
\end{figure}

\bpara{Motivation.} Despite significant algorithmic advances, existing CRT formulations rely on a critical assumption: the samples across channels must be time-aligned or equispaced, \ie, $\Tll = 0$. This requirement ensures that the CRT structure holds exactly. However, in many practical scenarios, this assumption is violated: (i) channel clocks may be imperfectly synchronized or exhibit drift and jitter, (ii) implementation of irregular sampling patterns (\eg multi-coset), or (iii) the system may operate asynchronously by design. In such cases, the usual CRT breaks down, and existing recovery guarantees no longer apply.
This significantly limits the applicability of CRT-based approaches in real-world sensing systems.

\bpara{Contribution.} In this work, we overcome this limitation by extending CRT-based \usf to asynchronous multi-channel sampling. 
We show that even when the time-varying offsets $\{ \Tll \}$ are non-uniform or irregular, exact recovery remains possible under mild conditions. 
Our analysis shows that the inherent smoothness of bandlimited functions induces consistency constraints across channels that can replace strict time alignment. 
By leveraging these constraints, we unify synchronous and asynchronous MC sensing within a common USF pipeline and generalize the CRT to asynchronous settings while retaining its desirable properties for high-dynamic-range recovery. This translates into two main advantages:

\noindent {\bf Wide applicability.} We broaden the applicability of CRT-based recovery to practical systems with clock mismatch or asynchronous acquisition.

\noindent {\bf Backwards Compatibility.} Our formulation remains backward compatible with existing multi-channel architectures such as multi-coset sampling (see \subfig{fig:demo}{b}{}), enabling CRT-style reconstruction in cases that were preciously out of reach.

\section{Asynchronized MC Unlimited Sensing }
\label{sec:MC}

\bpara{Problem Statement.} Let $\BL{\g}{\Omega}$ be an \HDR, continuous-time signal with maximum angular frequency $\Omega$. In asynchronous \usf pipeline, the multi-channel (MC) samples are given by\footnote{$\id{N} = \{0,1,\cdots,N-1\}$, where $N\in\Z^{+}$.},
\begin{equation}
\label{eq:samples}	
\yl{l} \sqb{n} = \MOL{\lambda_{l}}{\g(nT + \Tl{l})},  \quad n\in\id{N}, \quad l\in\id{L}
\end{equation}
where $\MOL{\lambda}{\cdot}$ denote the centered modulo operator defined as,
\begin{equation}
\label{eq:map}
\mathscr{M}_{\lambda_{l}}
:g \mapsto 2\lambda_{l} \left( {\fe{ {\frac{g}{{2\lambda_{l} }} + \frac{1}{2} } } - \frac{1}{2} } \right),
\begin{array}{l}
\ft{g} \DE g - \flr{g}\\
\lambda_{l}>0
\end{array}.
\end{equation}
In view of \eqref{eq:map}, the \HDR samples can be decomposed as, 
\begin{equation}
\label{eq:decom}
\gln{l} = \yln{l} + 2\lambda_{l} \gammaln{l}, \;\; \gammaln{l}\in\Z, \; \mbox{and}\;
\begin{array}{l}
\gln{l} = \g(nT + \Tl{l}) \\
\yln{l} = \yl{l} \sqb{n}
\end{array}.
\end{equation}
\emph{Goal}: Given the folded samples $\{\yln{l}\}_{n\in\id{N}}^{l\in\id{L}}$, the task is to recover the residue matrix $\mG = \sqb{\gammaln{l}}_{n\in\id{N}}^{l\in\id{L}}$.

\begin{figure*}[!t]
\centering	\includegraphics[width=\linewidth]{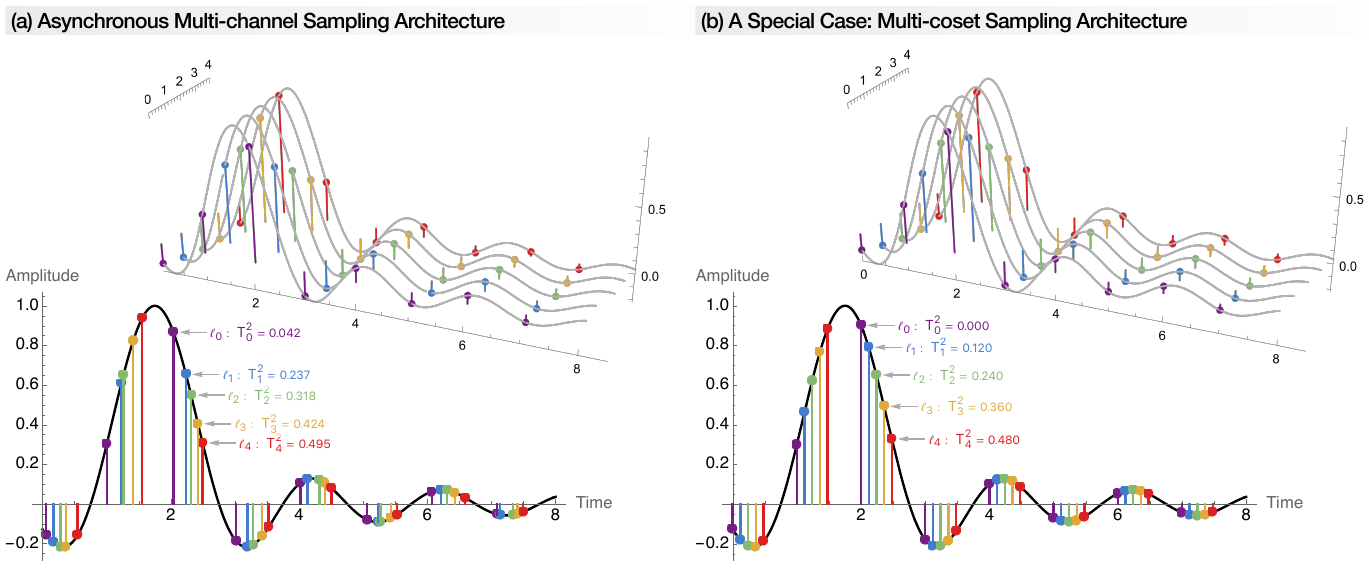}
\caption{
(a) Asynchronous MC sampling, which boils down to multi-coset architecture when $\Tl{l} = l\Td $. (b)  Despite the irregular setting of channel-offset $\Tl{l}$, the samples across channels at a fixed time instant exhibit certain smoothness (see slow-variation from $\Tl{0}$ to $\Tl{4}$ in both (a) and (b)). This channel redundancy imposes the smooth property of the complete graph constructed by $\{  \gln{l}\}_{l\in\id{L}}$ (also see \fig{fig:Asyn}).    
} 
\label{fig:demo}
\end{figure*}

Once $\mG$ is known, $\gln{l}$ can be recovered via \eqref{eq:decom} and thus, $\g\rob{t}$ can be reconstructed via standard bandlimited interpolation.
  
\bpara{Exact Recovery via Graph Modelling.} The recovery problem present in \eqref{eq:decom} amounts to a standard combinatorial optimization, which in general is messy and potentially intractable\footnote{Note that, the sampling step $T$ is not necessarily above the Nyquist rate. } as its computational complexity scales exponentially with the sample size $N$ and number of channels $L$.

To address this challenge, we employ a graph theory viewpoint. The multi-channel samples at a fixed time instant $n$ actually construct a complete graph with $L$ nodes. 
Let $\Gg = \rob{\{ \{ \gln{l} \}_{l\in\id{L}}\}, \mathrm{E}}$, where $\mathrm{E}$ is the edge set. The key insight in this paper is that, the bandlimitedness prior of $\g$ imposes smoothness constraint on $\Gg$—the quantity variation between arbitrary two nodes $\gln{l_{1}}, \gln{l_{2}}$ are \emph{bounded}. Mathematically, this graph smoothness condition can be algebraically expressed as,
\begin{equation}
\norm{ \B \mgn }_{\lp{\infty}} =
\mathop {\max}
\nolimits_{\substack{
l_{1}, l_{2}		
}}
\abs{  \gln{l_1} - \gln{l_2}  }
\leqslant \tauL, \quad \tauL > 0
\end{equation}
where $\mgn$ is the vector form of $\gln{l}$, and $\B\in\Z^{{(L-1)L}/{2}\times L}$ is its incidence matrix with $\B_{i,j} \in \{\pm1, 0\}, \forall i,j$ and $\norm{\B}_{\infty} = 2$. Here, $|{\sqb{ \B \mgn }_{j}}| = |\gln{l_1} - \gln{l_2}|$ measures the node quantity difference, and $\tauL$ characterizes the global smoothness of $\Gg$.

Inspired by the graph smoothness property, it can be proved that, the minimizer to the below integer programming,
\begin{equation}
\label{eq:min}	
\mathop {\min}
\nolimits_{\substack{
\tmgamman \in\Z^{L}		
}} \;  \norm{ \B \myn + 2\B\mL \tmgamman  }_{\lp{\infty}},\;\; \mL = \mathsf{diag}\rob{\{\lambda_{l}\}_{l\in\id{L}}}
\end{equation}
is \emph{unique} and \emph{exactly} coincides with the ground truth solution $\mgamman$, provided certain condition on graph smoothness $\tauL$.

\bpara{Constrained Graph Smoothness.} To quantify the interplay between signal prior and graph smoothness, we reorganize \eqref{eq:decom},
\begin{equation}
\label{eq:decom1}
\gl{\Lc}\sqb{n}  =  \rob{ \yl{l}  +   2\lambda_{l} \gammal{l}  + \Dg{l} } \sqb{n}, \quad
\begin{array}{l}
\Dg{l} = \gl{\Lc} - \gl{l}, \\
\gl{\Lc} = \g(nT + \Tc)
\end{array}	
\end{equation}
where $\Tc =(\max_{l} \Tl{l} + \min_{l} \Tl{l})/2$. 
In view of the definition of $\gln{l}$ above, we can derive that,
\begin{equation}
\abs{\Dg{l}\sqb{n}}	= \abs{ \int_{nT + \Tc}^{nT + \Tl{l}} \partial_{t} \g\rob{t}  dt  } \leqslant \abs{\Tc - \Tl{l}} \norm{ \partial_{t} \g }_{\Lp{\infty}}
\end{equation}
where $\partial_{t} \g$ is the derivative of $\g$. From Bernstein's inequality \cite{Bhandari:2020:Ja,Guo:2025:Ja}, we have $\BL{\g}{\Omega} \rightarrow \norm{ \partial_{t} \g }_{\Lp{\infty}} \leqslant \Omega \norm{\g}_{\Lp{\infty}}$, yielding,
\begin{equation}
\label{eq:BI}
\norm{ \Dg{l} }_{\lp{\infty}} \leqslant \abs{\Tc - \Tl{l}} \norm{ \partial_{t} \g }_{\Lp{\infty}} \leqslant \frac{1}{2}\Omega \tc \norm{\g}_{\Lp{\infty}}
\end{equation}
where $\tc = \max_{l_{1},l_{2}}\abs{ \Tl{l_{1}} - \Tl{l_{2}}}$. 
Consequently, plugging \eqref{eq:decom} and \eqref{eq:BI} into \eqref{eq:decom1}, we can deduce that,
\begin{equation*}
\norm{\wl{l}}_{\lp{\infty}} \leqslant \frac{1}{2} \Omega\tc \norm{\g}_{\Lp{\infty}}, \quad \forall l\in\id{L}
\end{equation*}
which suggests that, the maximum variation between arbitrary two nodes of $\Gg$ are upper bounded by,
\begin{equation}
\label{eq:c1}	
\norm{ \B \mgn }_{\lp{\infty}} \leqslant 2 \max_{l} \norm{\wl{l}}_{\lp{\infty}}
\leqslant  \Omega\tc \norm{\g}_{\Lp{\infty}} \DE \tauL.
\end{equation}
From \eqref{eq:c1}, the graph smoothness is governed by, (i) signal maximum frequency $\Omega$, (ii) input DR $\norm{\g}_{\Lp{\infty}}$ and (iii) maximal offset deviation $\tc$. Hence, the feasible solution set to \eqref{eq:min} can be achieved whenever a sufficient graph smoothness is guaranteed, which is stated in the following theorem:
\begin{figure*}[!t]
\centering	\includegraphics[width=0.9\linewidth]{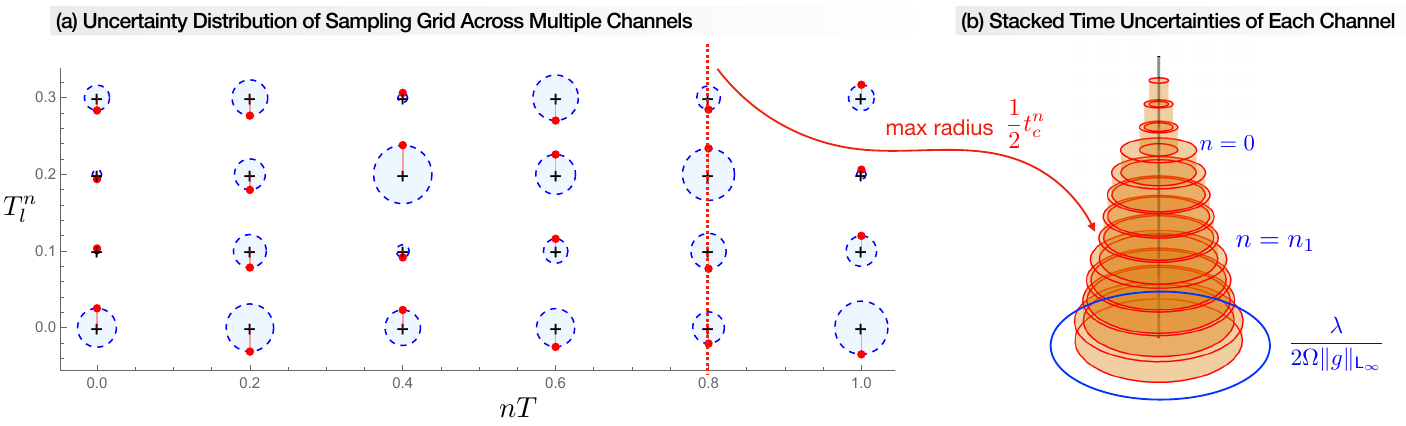}
\caption{
Visual illustration of exact graph recovery via smoothness constraint. (a) Uncertainty region of the sampling grid spanned by time-varying channel-dependent offsets $\Tl{l}$. (b) Stacked plot of the maximal sampling time offset across channels. A constrained graph smoothness (within the tolerance bound in (b)) guarantees the uniqueness and exactness of recovering $\Gg$ at each time instant-$n$ via constructed integer minimization in \eqref{eq:opt}, offering an efficient solution to the potentially intractable combinatorial optimization.
}
\label{fig:pattern}
\end{figure*}
\begin{theorem}
\label{thm:1}	
Let $\BL{\g}{\Omega}$ and the multi-channel folded measurements be $\yln{l} = \MOL{\lambda_{l}}{\g(nT + \Tl{l})}$. Then, $\{ \g(nT + \Tl{l})  \}_{n\in\id{N}}^{l\in\id{L}}$ can be exactly reconstructed, if $\lambda_{l} = \lam \kappal{l} $, $\lam\in\R^{+}$, $\{\kappal{l} \}_{l\in\id{L}}\in\Z_{+}^{L}$ are pair-wise co-prime, and 
\begin{equation}
\notag
\norm{\g}_{\Lp{\infty}}  < \min\rob{\tfrac{\lam}{\max\limits_{n} \max\limits_{l_{1},l_{2}}\abs{ \Tl{l_{1}} - \Tl{l_{2}}}  \Omega}, \lam \prod\limits_{l\in\id{L}} \kappal{l}  - \max_{l} \lambda_{l} }.
\end{equation}
\end{theorem}
\begin{proof}
We prove the theorem by showing the uniqueness and exactness of the integer solution. By definition, $\gcd{\{\kappal{l}\}_{l\in\id{L}}} = 1$ and $\glcm =\lambda \lcm{\{\kappal{l}\}_{l\in\id{L}}}$, where $\lcm{\cdot}$ and $\gcd{\cdot}$ denote the least common multiple and greatest common divisor operators, respectively. Then we have,
\begin{equation}
\label{eq:lcm}	
\gcd{\{ \kappal{l} \}_{l\in\id{L}}} = 1 \quad \mathsf{and} \quad \glcm = \lam \prod\nolimits_{l\in\id{L}} \kappal{l} .
\end{equation}
From the hypothesis, we can deduce that, 
\begin{equation}
\norm{\g}_{\Lp{\infty}}  < \lam \lcm{\{\kappal{l}\}_{l\in\id{L}}}  - \max\nolimits_{l} \lambda_{l}
\end{equation}
which suggests, $\max_{l} \QO{\lambda_{l}}{\norm{\g}_{\Lp{\infty}}} < \glcm$, where $\QO{\lambda}{\cdot} = 
2\lambda \flr{ \tfrac{ (\cdot) }{2\lambda } + \tfrac{1}{2} }
$. Consequently, we can write that,
\begin{equation}
\label{eq:c2}	
\gammal{l}\sqb{n} \in \sqb{-\Gammal{l} ,  \Gammal{l}  }, \quad \Gammal{l} =  \flr{\tfrac{\norm{\g}_{\Lp{\infty}} + \lambda_{l} }{2\lambda_{l}}} < \frac{\glcm}{2\lambda_{l}}.
\end{equation}
With \eqref{eq:c1} and \eqref{eq:c2}, the minimization in \eqref{eq:min} reduces to,
\begin{equation}
\label{eq:opt}	
\mathop {\min}
_{\substack{
\tmgamman
}} \ 	\norm{\B \myn + 2\B\mL \tmgamman}_{\lp{\infty}},\; \mbox{s.t.} \; \abs{\sqb{\tmgamman}_{l} } <  \tfrac{\glcm}{2\lambda_{l}}, \;\; l\in\id{L}.
\end{equation}
To evaluate the exactness and uniqueness of the solution space,
let $\tmgamman$ be the minimizer to \eqref{eq:opt}, and by definition, we have,
\begin{equation}
\label{eq:B1}	
\norm{\B \myn + 2\B\mL \tmgamman }_{\lp{\infty}} \leqslant \norm{\B \myn + 2\B\mL \mgamman }_{\lp{\infty}} \stackrel{\eqref{eq:c1}}{\leqslant} \tauL.
\end{equation}
Thus, evaluating the deviation $\B\mL (\tmgamman - \mgamman)$ leads to,
\begin{equation}
\label{eq:B2}	
\B\mL (\mgamman - \tmgamman ) = \frac{1}{2} (\B \mgn -  \B \tmgn) , \quad \tmgn = \myn + 2\mL \tmgamman.
\end{equation}
Let $\mK = \mathsf{diag}\rob{\{ \kappal{l}  \}_{l\in\id{L}}}$. Combining \eqref{eq:B1} and \eqref{eq:B2} results in,
\begin{equation}
\norm{\B\mK(\mgamman - \tmgamman)}_{\lp{\infty}} \leqslant \frac{\tauL}{\lam} \quad \mbox{where} \quad \mL \stackrel{\eqref{eq:lcm}}{=} \lam \mK .
\end{equation}
Moreover, the hypothesis on $\norm{\g}_{\Lp{\infty}}$ indicates that,
\begin{equation*}
\max_{n}\max_{l_{1},l_{2}}\abs{ \Tl{l_{1}} - \Tl{l_{2}}}  \Omega \norm{\g}_{\Lp{\infty}}  < {\gcd{\{\lambda_{l}\}_{l\in\id{L}}}} \Longrightarrow \tauL < \lam.
\end{equation*}
And therefore, we can deduce that,
\begin{equation}
\norm{\B\mK(\mgamman - \tmgamman)}_{\lp{\infty}} < 1 \Longrightarrow  \B\mK(\mgamman - \tmgamman) = \mathbf{0}
\end{equation}
which follows from the integer constraints of $\{ \B,\mK,\mgamman,\tmgamman \}$. Since $\B$ is induced by a complete graph, then we can write,
\begin{equation}
\label{eq:B3}	
\mK(\mgamman - \tmgamman)  = \hn \glcm \mathbf{1}, \quad \hn\in\Z
\end{equation}
where $\mathbf{1}$ is an all $1$ vector that spans the null space of $\B$. Note that, \eqref{eq:B3} has a common scaling factor $\glcm$ since the weighted variation between arbitrary two nodes is zero, \ie, $\B\mK(\mgamman - \tmgamman) = \mathbf{0}$.  
Combining with bounded integer range constraint in \eqref{eq:c2}, we can conclude that,
\begin{equation}
\norm{\mgamman - \tmgamman}_{\lp{\infty}} < \max_{l} \frac{\glcm}{\lambda_{l}}  
\Longrightarrow 
\hn  < 1
\end{equation}
implying, $\hn = 0$. This shows that the minimizer $\tmgamman$ to \eqref{eq:opt} is \emph{unique} and \emph{coincides with the ground truth solution $\mgamman$}.
A visual illustration of graph recovery is provided in \fig{fig:pattern}.
\end{proof}
\bpara{Remarks.} The key takeaways from \thmref{thm:1} are threefold:
\begin{enumerate}[leftmargin =*, label = {R\arabic*}),labelsep = 5pt]
\item The least common multiple and greatest common divisor characterizes the maximal dynamic range and error tolerance, respectively. Increasing the number of channels and channel-offset resolution could expand the achievable DR and enhance the system robustness. 
\item The sampling conditions in \thmref{thm:1} are independent of the time step $T$, allowing for high-resolution audio reconstruction from low-resolution measurements. 
\item \thmref{thm:1} could be extended to the quantized measurements, where the smoothness of $\Gg$ is upper bounded by, $2 \max_{l} \norm{\wl{l}}_{\lp{\infty}} \leqslant \Omega \tc \norm{\g}_{\Lp{\infty}} +\max_{l_{1},l_{2}} \norm{\qnl{l_{1}}}_{\lp{\infty}} + \norm{\qnl{l_{2}}}_{\lp{\infty}}$, and $\qnl{l}$ is the quantization noise of node-$l$.
\item \thmref{thm:1} guarantees the uniqueness and exactness of the integer solution, enabling the efficient pairwise retrieval of $\{\gammal{l_{1}}\sqb{n}, \gammal{l_{2}}\sqb{n}\}$ via look-up table.
\end{enumerate}	

\section{High-resolution Signal Recovery}

The proof of \thmref{thm:1} is constructive and leads to an efficient and accurate signal reconstruction approach. Concretely, our method comprises two steps: (i) retrieval of $\{\g\rob{nT + \Tl{l} } \}_{n\in\id{N}}^{l\in\id{L}}$ via solving the integer minimization in \eqref{eq:opt} and (ii) estimation of $\{ \g\rob{ n\Td } \}_{n\in\id{KN}}$ with $\Td = T/K, K\in\Z^{+}$ via linear signal interpolation.

\begin{table}[!htb]
\centering
\caption{Recovery Benchmarking: Parameters and Performance Metrics.}
\label{tab:exp}
\resizebox{0.8\textwidth}{!}{%
\rowcolors{3}{rowgray}{white}
\begin{tabular}{@{}ccccccccccc@{}}
\toprule
\multicolumn{1}{c}{Dataset} 
& $\norm{\g}_{\Lp{\infty}}$
& \multicolumn{2}{c}{Thresholds}
& \multicolumn{1}{c}{$\frac{\Omega}{2\pi}$}
& \multicolumn{1}{c}{$\frac{1}{T}$}
& \multicolumn{1}{c}{Run Time}
& \multicolumn{1}{c}{MSE ($\times 10^{-3}$)} \\
\cmidrule(lr){3-4}
& 
& $\lambda_{1}$ & $\lambda_{2}$
& \khz
& \khz
& second 
& 
&  \\ \midrule

\rom{1}&$14.86$&$1.50$&$5.50$&$4.00$&$42.67$&$4.18$&$0.03$\\
\rom{2}&$6.27$&$1.00$&$3.50$&$6.40$&$25.60$&$3.15$&$0.46$\\
\rom{3}&$13.39$&$1.50$&$5.00$&$5.05$&$19.20$&$2.66$&$0.21$\\
\rom{4}&$7.93$&$1.00$&$4.50$&$5.49$&$12.80$&$2.25$&$0.79$\\
\rom{5} (\fig{fig:exp})&$15.42$&$2.50$&$3.50$&$3.20$&$6.40$&$1.97$&$27.00$\\

\bottomrule
\end{tabular}
}
\end{table}

For step-2, advanced techniques like spline or wavelet interpolators, designed for signal inpainting/interpolation could be utilized in our context \cite{Blu:2001:J}; while, here we use a simple linear interpolation approach for computational efficiency purpose, enabling real-time processing of million-scale measurements.

\section{Experimental Validation}
\label{sec:exp}

We conduct the following numerical experiments to show the performance gains of our asynchronous MC \usf scheme.

\begin{figure}[!t]
\centering	\includegraphics[width=0.8\linewidth]{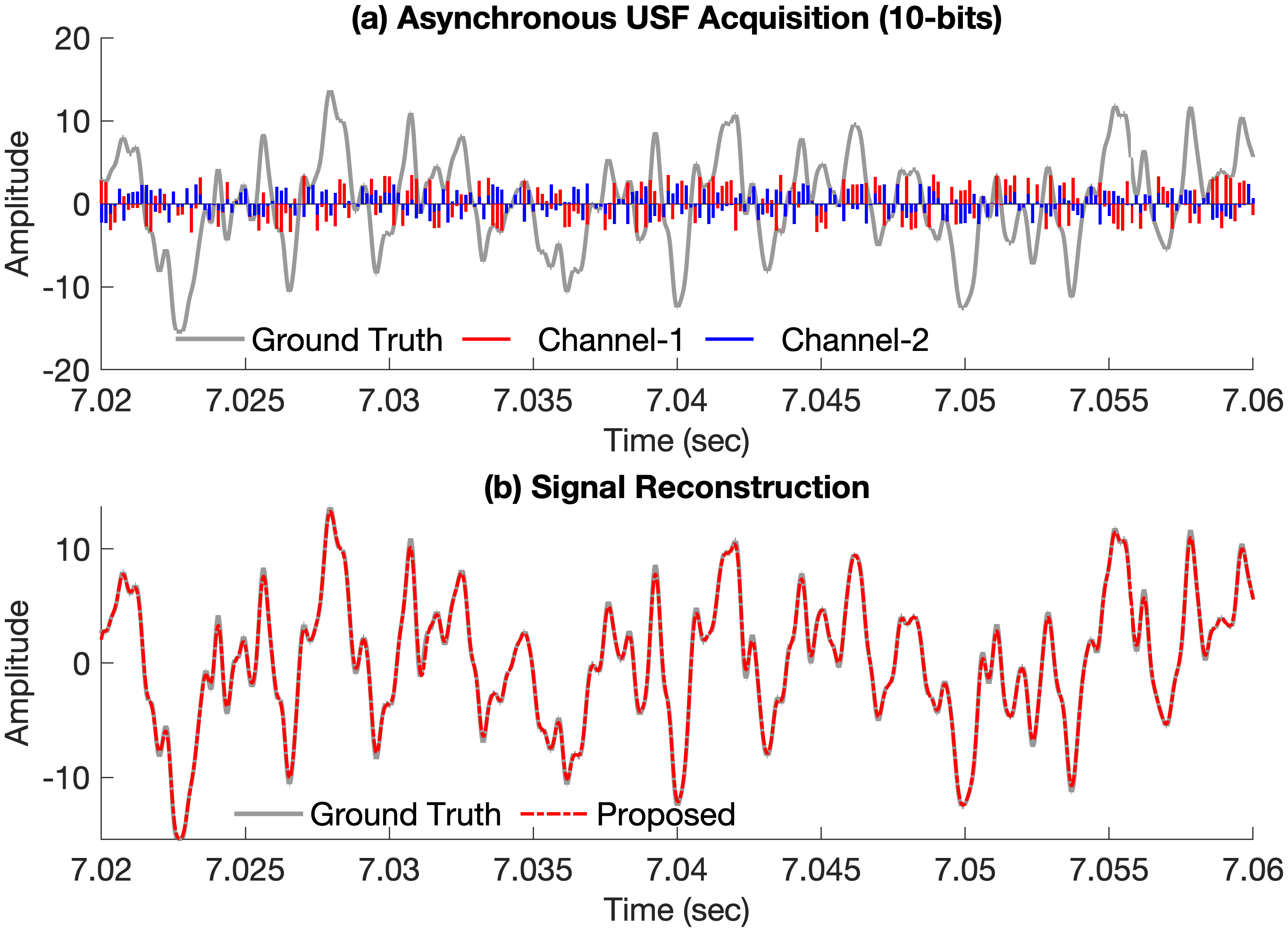}
\caption{
Signal recovery from asynchronous MC modulo measurements. A dual-channel multi-coset architecture is used to capture the input signal near Nyquist-rate with $10$-bit quantization, as shown in (a). Despite the million-scale data challenge and quantization noise, our method offers an efficient, high-precision signal recovery, validating the results in \thmref{thm:1}.
} 
\label{fig:exp}
\end{figure}

\bpara{Experimental Protocol.} We randomly generate the bandlimited signals with maximal frequency ranging from $3.2$ \khz to $6.5$ \khz. A dual-channel multi-coset architecture (see \subfig{fig:demo}{b}{}) is utilized to capture the input signal $\g$ with $10$-bit quantization, $\Tl{l} = l\Td$ and $\Td$ = 3 $\mu$s, resulting in approximately $4\times10^5$ samples in all cases.
MSE (mean-squared error) is used to evaluate the signal recovery precision.
We progressively reduce the sampling rate to assess the performance limit of our method.
Experimental parameters like folding thresholds, sampling frequency, and input DR are tabulated in \tabref{tab:exp}.

As shown in \tabref{tab:exp}, in all scenarios, our method recovers the samples at $\{nT + l\Td\}_{n,l}$ down to the precision of the modulo-ADCs, \ie, $\max_{l} \lambda_{l}/(2^B - 1)$, effectively validating the results in \thmref{thm:1} and robustness against quantization. In medium oversampling regime, the proposed interpolation approach achieves accurate signal estimates ($\mathrm{MSE} < 5\times 10^{-4}$) up to $K=20$ times upsampling.
When pushing close to the Nyquist rate (last two rows in \tabref{tab:exp}), the temporal sample correlation diminishes, creating challenges for oversampling-based methods, such as non-linear filtering \cite{Bhandari:2020:Ja} and Fourier-partitioning \cite{Bhandari:2021:J}. Despite the voluminous data and low sampling rate, our method offers reasonable signal reconstruction with significant dynamic range expansion, as shown in \fig{fig:exp}.

\section{Conclusion}

Unlimited Sensing Framework is an emerging digitization paradigm that delivers concurrent high-dynamic-range and high-digital-resolution, beyond the capacity of conventional sampling schemes.
In this paper, we extend \usf to asynchronous multi-channel architectures where each channel samples contain irregular time offset.
Harnessing channel redundancy, we show that constrained smoothness across channels results in an exact and unique solution to the constructed integer minimization. This enables an efficient method to solve the potentially intractable combinatorial optimization, which is validated by concrete numerical experiments.

\bibliographystyle{IEEEtran_URL}

\end{document}